\newif\ifshort
\shortfalse
\documentclass[11pt]{article}
\usepackage[a4paper, total={16cm, 24cm}]{geometry}
\date{}
\usepackage{natbib}
\renewcommand\tableofcontents{\listoftoc*{toc}} 

\usepackage{authblk}
\author[1]{Piotr Faliszewski}
\author[1]{Łukasz Janeczko}
\author[2]{Grzegorz Lisowski}
\author[1]{Kristýna Pekárková}
\author[3,4]{Ildikó~Schlotter}

\affil[1]{AGH University, Poland} 
\affil[2]{University of Groningen, The Netherlands} 
\affil[3]{ELTE Centre for Economic and Regional Studies, Hungary} 
\affil[4]{Budapest University of Technology and Economics, Hungary}

\usepackage{times}  
\usepackage{helvet}  
\usepackage{courier}  
\usepackage[hyphens]{url}  
\usepackage{graphicx} 
\usepackage{natbib}  
\usepackage{caption} 
\usepackage{algorithm}
\usepackage{algorithmic}

\usepackage{newfloat}
\usepackage{listings}
\DeclareCaptionStyle{ruled}{labelfont=normalfont,labelsep=colon,strut=off} 
\floatstyle{ruled}
\newfloat{listing}{tb}{lst}{}
\floatname{listing}{Listing}

\title{Identifying Imperfect Clones in Elections}

\usepackage{graphicx}
\usepackage{amsthm}
\usepackage{amsmath}
\usepackage{amssymb}
\usepackage{amsfonts}
\usepackage{epsfig}
\usepackage{mathtools}
\usepackage{thmtools}
\usepackage{thm-restate}
\usepackage{enumitem}
\usepackage{stackengine,graphicx}
\usepackage{xspace}
\usepackage{float}
\usepackage{xcolor}
\usepackage{booktabs}
\usepackage{nicefrac}
\usepackage{multirow}
\usepackage[title]{appendix}

\newcommand{\approxclones}{\textsc{Approximate-Clones Partition}}
\newcommand{\findqclone}{\textsc{Find-$q$-Approximate-Clone}}
\newcommand{\manyperfectclones}{\textsc{Independent-Clones Partition}}
\newcommand{\GlobalSubelectionPart}{\textsc{Subelection-Clones Partition}}

\newcommand{\coverproblem}{\textsc{Exact-3-Cover}}
\newcommand{\triangulation}{\textsc{Partition into Triangles}}

\def\np{\mathsf{NP}}
\def\paranp{\mathsf{para}$-$\mathsf{NP}}
\def\p{\mathsf{P}}
\def\procClean{\mathsf{BreakSegment}}
\def\qbef{\hookleftarrow}
\def\qaft{\hookrightarrow}
\def\qnorm{\circ}
\def\T{\mathcal{T}}
\def\Tcomp{\mathcal{T}^2_{\mathsf{comp}}}
\def\dummyseries{[\circ \circ \circ]}
\newcommand{\ora}[1]{\overrightarrow{#1}}
\newcommand{\pref}{\succ}
\newcommand{\size}{\mathrm{size}}
\newcommand{\calC}{{\mathcal{C}}}

\newcommand{\calR}{{\mathcal{R}}}
\newcommand{\calS}{{\mathcal{S}}}

\newcommand{\shortcitethm}[1]{{\small (T\ref{#1})}}
\newcommand{\shortciteprop}[1]{{\small (P\ref{#1})}}

\newcommand{\pos}{{\mathrm{pos}}}
\newcommand{\rank}{{\mathrm{pos}}}
\newcommand{\shift}{{\mathrm{sh}}}
\newcommand{\maybe}[1]{{\underline{#1}}}

\newtheorem{theorem}{Theorem}[section]
\newtheorem{corollary}[theorem]{Corollary}

\newtheorem{proposition}[theorem]{Proposition}

\newtheorem{definition}[theorem]{Definition}
\newtheorem{remark}[theorem]{Remark}

\newcommand{\hy}{\hbox{-}\nobreak\hskip0pt}

\newcommand{\FPT}{\ensuremath{\mathsf{FPT}}\xspace}
\newcommand{\fpt}{\FPT}
\newcommand{\XP}{\ensuremath{\mathsf{XP}}\xspace}
\newcommand{\xp}{\XP}

\newcommand{\Wh}[1]{$\mathsf{W[#1]}$\hy{}hard\xspace}

\allowdisplaybreaks

\usepackage{todonotes}
\usepackage{color}
\usepackage{cleveref}

\begin{document}

\maketitle

\begin{abstract}
  A perfect clone in an ordinal election (i.e., an election where the
  voters rank the candidates in a strict linear order) is a set of candidates that each voter
  ranks consecutively. We consider different relaxations of this notion: 
  \emph{independent} or \emph{subelection clones} are sets of candidates that only some of the
  voters recognize as a perfect clone, whereas \emph{approximate clones} are
  sets of candidates such that every voter ranks their members close to
  each other, but not necessarily consecutively. We establish the complexity of identifying such imperfect clones, and
  of partitioning the candidates into families of imperfect clones. 
  We also study the parameterized complexity of these problems with respect to a set of natural parameters such as the number of voters, the size or the number of imperfect clones we are searching for, or their level of  imperfection.
\end{abstract}

\section{Introduction}

Let us consider an ordinal election, with a given set of candidates
and a collection of voters that rank these candidates. Such elections
may appear in 
the classic political settings---e.g., when the goal is to elect a
president or some other leader---but also in various other contexts,
such as sport events---e.g., where the candidates are the athletes and
the ``voters'' are particular competitions, ranking them according to
their performance---or multicriteria evaluations---e.g., when listing
the best universities or the most livable cities, where the ``voters''
are the applied quality measures; for examples of such election data,
see the work of \citet{boe-sch:c:election-data}.  In each of these
settings, we expect some groups of candidates to be similar to each
other. For example, left- or right-wing politicians are likely to have
similar platforms, athletes with similar abilities---such as, e.g.,
the top sprinters in Tour de France---typically perform similarly, and
universities of a particular type---such as small liberal arts
colleges---have similar features. Consequently, we expect voters to
rank such candidates close to each other.  Formally, sets of
candidates that all voters rank on consecutive positions are called
\emph{clones} and were already studied in some detail both in AI and
(computational) social choice; as examples, we point to the works of
\citet{tid:j:clones}, \citet{las:j:rank-based,las:j:aggregation} and
\citet{elk-fal-sli:j:cloning,elk-fal-sli:c:decloning}; see also the
notion of \emph{composition
  consistency}~\citep{laf-lai-las:j:composition,las:b:tournament-solutions}
and its connection to
clones~\citep{ber-cas-rob-ong-con-elk:c:clones}. Yet, \emph{similar}
does not mean \emph{identical} and, so, expecting all the voters to
perfectly recognize all the clones---i.e., rank them
consecutively---seems overly demanding.  Hence, we focus on the
problem of identifying those candidates that form \emph{imperfect
  clones}, and on the problem of partitioning the set of candidates
into such clones. We believe that the ability to solve these problems
allows one to better understand the structure and nature of the
candidates in a given election; in this sense, we follow up on the
work of \citet{jan-lan-lis-szu:c:consistent-subelections} on
identifying hidden structures in voters' preferences.
Our work is also related to that of
\citet{pro-sch-zha:c:approximate-clones}, who consider approximate
clones in the context of training LLMs.
However, in contrast to their work, our focus is on establishing the
(parameterized) complexity of our problems.
%

There are two main ways in which clones can be imperfect. First, we
may require all the voters to recognize all the clones, but cut them
some slack on ranking their members consecutively: Voters simply need
to rank clone members on close enough positions, but they are allowed
to put some candidates that are not members of the given clone in
between. Second, we may require the voters that recognize the clones
to indeed rank their members consecutively, but allow each clone to be
recognized by a possibly different subset of the voters. We refer to
the former type of clones as \emph{approximate} and to the latter one
as \emph{independent}. If we insist that all the independent clones
are recognized by the same voters, then we refer to them as
\emph{subelection} clones. We find that the complexity of our clone
identification/clone partition problems very strongly depends on the
type of imperfection that we consider. For example, we obtain the
following results:
\begin{enumerate}
\item There is a simple, polynomial-time algorithm for identifying
  independent clones, 
  but recognizing approximate clones is $\np$-hard. Yet, if members of
  approximate clones are not ranked too far apart, then we can
  identify them using an $\fpt$ algorithm (see
  \Cref{sec:clone-structures}).

\item Partitioning the set of candidates into a given number of
  imperfect clones, each of at most a given size, is $\np$-hard, but
  there are algorithms for some special cases. For example, there is
  an $\fpt$ algorithm for approximate clones for parameterization by
  their imperfection level, an $\xp$ algorithm for
  independent/subelection clones for parameterization by the number of
  clones, and an $\fpt$ algorithm for subelection clones for
  parameterization by the number of voters (see \Cref{sec:partitions}
  and Table~\ref{tab:results}).
\end{enumerate}
Indeed, we classify the parameterized complexity of the partitioning
problem for all our basic parameters that clone partitions may have.
We defer some of our proofs to
\ifshort the 
supplementary material
\else
the appendix 
\fi
and mark such  results with the~$\star$ symbol.

\section{Preliminaries}

For an integer $k$, we write $[k]$ for the set $\{1, \ldots, k\}$.

\paragraph{Elections.}
An \emph{(ordinal) election} is a pair $E = (C,V)$, where
$C = \{c_1, \ldots, c_m\}$ is a set of \emph{candidates} and
$V = \{v_1, \ldots, v_n\}$ is a set
of \emph{voters} (sometimes
also referred to as \emph{votes}). 
Each voter $v$ has a \emph{preference
  order} $\pref_{v}$ that strictly ranks all candidates in~$C$.
  We write $v \colon a \pref b$ to indicate that $v$ ranks candidate $a$
above candidate~$b$, and we use analogous notation for larger
candidate sets. A set of candidates in a preference
order should be interpreted as ranking its members in some arbitrary but fixed
order.
For a voter $v$ and candidate $c$, we write
$\pos_{v}(c)$ to denote the position of~$c$ in $v$'s preference
order; the top-ranked candidate has position~$1$.

\paragraph{Clone Structures and Clone Partitions.}Let $E = (C,V)$ be
an election and let $b \geq 1$ be an integer. We say that a nonempty
subset $X$ of candidates is a $b$-\emph{perfect clone} in $E$ if
$|X| \leq b$ and each voter ranks members of $X$ consecutively (but,
possibly, in a different order).
We are also interested in clones that are in some way imperfect, i.e.,
either not all voters recognize the given candidates as forming a clone
(so \emph{independent} groups of voters may recognize different
clones), or the voters do not agree on the exact compositions of the
clones (so the clones are \emph{approximate}). We formalize these
ideas below:

\begin{description}
\item[Approximate Clones.] Let $b \geq 1$ and $q \geq 0$ be
  integers. A subset $X$ of candidates is a
  \emph{$(q,b)$-approximate clone} if $|X| \leq b$ and for each voter
  $v \in V$ we have
  $\max_{x \in X} \pos_v(x) - \min_{x \in X} \pos_v(x) \leq q+b-1$.
  In other words, each voter ranks the candidates from $X$
  \emph{almost} consecutively, ranking at most $q$ 
  candidates outside~$X$ in between any two candidates from~$X$. 
\item[Independent Clones.] Let $b$ and~$w$ be positive integers. A subset $X$ of candidates is a $(w,b)$-\emph{independent
    clone} if $|X| \leq b$ and there are at least $w$ voters who \emph{recognize} $X$, i.e., rank all
  members of~$X$ consecutively.
\end{description}

A $q$-approximate clone or a $w$-independent clone is a clone that is
$(q,b)$-approximate or $(w,b)$-independent for some~$b$, respectively.  A
\emph{\{perfect, $q$-approximate, $w$-inde\-pen\-dent\}-clone structure} for
election $E$ is the set of all \{perfect, $q$-approximate,
$w$-inde\-pendent\}-clones that appear in~$E$.  A \emph{\{perfect, $q$-approximate,
  $w$-independent\}-clone partition} for $E$ is a family
$\calC = \{C_1, \ldots, C_t\}$ of subsets of~$C$ whose
members are \{perfect, $q$-approximate, $w$-independent\}-clones and each
candidate belongs to exactly one~$C_i$. We also consider
subelection-clone partitions: 
\begin{description}  
\item[Subelection Clones.] Let $b$ and~$w$ be positive integers. Then
  $(C_1, \ldots, C_t)$ is a \emph{subelection-clone partition} if
  there is a subset $W \subseteq V$ of at least $w$ voters such that
  $(C_1, \ldots, C_t)$ is a perfect-clone partition for election~$(C,W)$ with each clone~$C_i$ having size at most~$b$. We refer to
  the members of such a partition as \emph{$(w,b)$-subelection clones}.
\end{description}

We often speak of clones, clone structures, and clone partitions
without specifying the exact type of clones involved. In such cases,
the relevant type is clear from the context. Similarly, we sometimes
speak of, e.g., $(q,b)$-approximate-clone partition to mean an
approximate-clone partition whose each member is a $(q,b)$-approximate
clone.

\begin{remark}
  The difference between independent clones and
  subelection clones is that 
  two sets are
  $(w,b)$-independent clones if for each of them there is a group of
  $w$ voters who ranks their members consecutively, but these two
  groups do not need to coincide and for small enough $w$ may even be
  disjoint.
  On the other hand, subelection clones have to be ranked
  consecutively by the same $w$ voters.
\end{remark}

Our theoretical results mostly regard the computational complexity of
the following family of problems.

\medskip
\noindent
\begin{minipage}{0.88\columnwidth}
\fbox{
    \begin{tabular}{@{\hspace{2pt}}l@{\hspace{4pt}}p{0.84\columnwidth}@{\hspace{2pt}}}%
    \multicolumn{2}{@{\hspace{2pt}}p{0.87\columnwidth}}{\textsc{\{Perfect, Approximate, Independent, 
    Subelection\}-Clone Partition}} 
    \\[3pt] 
        {\bf Input:}  
         & An election
  $E = (C,V)$, positive integers $b$  and~$t$, as well as integer
  $q \geq 0$ for the \textsc{Approximate} variant, and integer
  $w \geq 1$ for \textsc{Subelection} and \textsc{Independent}
  variants. \\
         {\bf Question:} & Is there a  clone partition of size at most~$t$ for election~$E$ 
  that consists of $b$-perfect, $(q,b)$-approximate,
  $(w,b)$-independent, or $(w,b)$-subelection clones, respectively?
    \end{tabular}
    }
\end{minipage}
\medskip

\paragraph{Computational Complexity.}
We assume familiarity with classic and parameterized computational
complexity theory, as presented in the textbooks of
\citet{pap:b:complexity}, \citet{nie:b:invitation-fpt} and
\citet{cyg-fom-kow-lok-mar-pil-pil-sau:b:fpt}. 
The following problem plays an important role in our hardness proofs.

\begin{definition}
  The input of the \textsc{Restricted Exact Cover by 3-Sets (RX3C)} problem, consists of a ground set~$X$ and a family $\calS$
  of size-$3$ subsets of~$X$ with each $x \in X$ occurring in exactly three sets of~$\calS$. We ask if there is a collection of $k=\nicefrac{|X|}{3}$ sets
  from~$\calS$ whose union is $X$.
\end{definition}
\noindent
Naturally, the $k$ sets about whose existence we ask must be disjoint.
This variant of the problem was shown to be $\np$-complete by \citet{gon:j:x3c}. Some of our proofs will rely on the 7-colorability of a certain graph underlying an instance of RXC3, as stated in the following observation.

\begin{proposition}
    \label{prop:RX3C-7colorable}
    Given an instance $(X,\calS)$ of RX3C, we can partition~$\calS$ in polynomial time into $\calS_1,\dots,\calS_7$ so that for each $i \in [7]$ the sets in $\calS_i$ are pairwise disjoint.
\end{proposition}
\begin{proof}
    Consider the graph~$H$ whose vertex set is~$\calS$ and where two triples from~$\calS$ are connected by an edge if and only if they share a common element of~$X$. The maximum degree in~$H$ is at most~$6$, because for any $S \in \calS$, there are at most two triples in~$\calS \setminus \{S\}$ containing a fixed element~$x \in S$.
    Thus, we can find a proper coloring of~$H$ with~$7$ colors in polynomial time (using a greedy coloring). The 7 color classes obtained give the desired partitioning of~$\calS$.
\end{proof}

We will also use the following special variant of the \textsc{Exact Cover}
problem in some of our algorithms.

\begin{definition}
  In the \textsc{Weighted Exact Cover (WEC)} problem we are given a set $X = [m]$ of
  elements, a family $\calS = \{S_1, \ldots, S_n\}$ of subsets of $X$
  where each set $S_i$ is associated with weight $w_i$, and a
  nonnegative integer $t$. We ask if there is a set $I \subseteq [n]$
  such that 
  \begin{enumerate}
    \item for each $a, b \in I$, $a \neq b$, $S_a \cap S_b = \emptyset$, 
    \item $\bigcup_{\ell \in I}S_\ell = X$, and 
    \item $\sum_{\ell \in I} w_\ell \leq t$.
  \end{enumerate}
\end{definition}

Consider some positive integer $m$. 
A set $A \subseteq [m]$ is an \emph{interval} if either $A=\emptyset$ or
$A = \{i, i+1, i+2, \ldots, j\}$ for some  integers $i$ and $j$, $i \leq j$. Further, given a nonnegative integer~$q$, we say that $A$ is a \emph{$q$-approximate interval} if the set
\begin{align*}
  \{  \min&(A)+1, \min(A)+2, \ldots, \min(A)+(q-1)\}
  \cup A \, \cup \\
  &\{ \max(A)-1, \max(A)-2, \ldots, \max(A)-(q-1)\}
\end{align*}
is an interval. Roughly speaking, a $q$-approximate interval is an interval that,
possibly, is missing some of the first and last~$q$ values. The next result 
follows due to a simple dynamic programming, provided for the sake of completeness.

\begin{proposition}\label{pro:wec-q-approx}
  There is an algorithm that given an instance~$I$ of \textsc{WEC} where
  each set is a $q$-approximate interval decides in time  $\mathcal{O}^*(2^{4q})$ if
  $I$ is a yes-instance.\footnote{The notation $\mathcal{O}^*$ hides polynomial factors.}
\end{proposition}
\begin{proof}
  Consider an instance of \textsc{WEC} with ground set $[m]$, a family
  $\calS = \{S_1, \ldots, S_n\}$ of sets where each $S_i$ is a
  $q$-approximate interval of weight $w_i$, and an integer $t$.  We
  define a function $f$ such that, for each integer
  $i \in [m] \cup \{0\}$ and set
  $B \subseteq \{i+1, i+2 \ldots, i+2q\}$, $f(i,B)$ is the smallest
  number $t'$ for which there is a set $I' \subseteq [n]$ such that:
  \begin{enumerate}
  \item for each $a, b \in I'$, $a \neq b$,
    $S_a \cap S_b = \emptyset$,
  \item $\bigcup_{\ell \in I'}S_\ell = [i] \cup B$, and
  \item $\sum_{\ell \in I'}w_\ell = t'$.
  \end{enumerate}
  If these conditions cannot be met, then we let ${f(i,B) = \infty}$.  As
  a base case, we observe that $f(0,\emptyset) = 0$.  Consider
  integers $i, j \in [m]$, $i \leq j$, sets
  $B' \subseteq \{i+1, \ldots, i+2q\}$ and
  $B'' = \{j+1, \ldots, j+2q\}$, and a set $S_\ell$. We say that
  $(i,j,B',B'', S_\ell)$ is \emph{consistent} if:
  \begin{enumerate}
  \item $S_\ell \cap ([i] \cup B') = \emptyset$, and
  \item $S_\ell \cup ([i] \cup B') = [j] \cup B''$.
  \end{enumerate}
  In other words, if $(i,j,B',B'',S_\ell)$ is consistent and we have a
  family of sets whose union is $[i] \cup B'$, then we can extend this
  family with set $S_\ell$ so the union of the extended family is
  $[j] \cup B''$. Now we express $f(i,B)$ recursively as follows (we
  consider $j \in [m]$ and $B'' \subseteq \{j+1, \ldots, j+2q\}$):
  \begin{align*}
    f(j,B'') = 
    \min_{\substack{(i,j,B',B'',S_\ell)\\ \text{is consistent}}} f(i,B') + w_\ell,
  \end{align*}
  and we let $f(j,B'') = \infty$ if there is no consistent $(i,j,B',$
  $B'',S_\ell)$ to consider.  Using this recursive formulation and
  standard dynamic programming techniques, we can compute
  $f(m,\emptyset)$ in time $\mathcal{O}^*(2^{4q})$; this running time stems from
  the fact that there are $\mathcal{O}^*(2^{2q})$ arguments that $f$ can take, and
  our recursive formula requires us to consider $\mathcal{O}^*(2^{2q})$ consistent
  tuples.  We accept if $f(m,\emptyset) \leq t$.  The algorithm is
  correct because $\calS$ consists of $q$-approximate intervals.
\end{proof}

\section{Complexity of Finding Clone Structures}\label{sec:clone-structures}

For the case of perfect clones, \citet{elk-fal-sli:c:decloning}
characterized what set families can appear as perfect-clone structures
and provided simple algorithms for computing such structures. Similar
algorithms also work for independent clones, yielding the
following. 
\begin{proposition}
  There are polynomial-time algorithms that, given an election~$E$ and an
  integer~$w$, (a)~test if a given set of candidates forms an
  independent clone recognized by at least $w$ voters, and (b)~compute
  a clone structure consisting of such independent clones.
\end{proposition}
\begin{proof}
To compute an independent-clone
structure for a given election $E=(C,V)$ where each clone has to be recognized by at least
$w$ voters, we check 
for each voter $v \in V$ and each $i, j$ such that
$1 \leq i \leq j \leq |C|$ whether the candidates that $v$
ranks at positions between~$i$ and~$j$ form an independent clone; to
do so, we look at all the votes and check if at least $w$ of them rank
these candidates consecutively. 
\end{proof}

In contrast, deciding if there is a single $q$-approximate
clone of a certain size is $\np$-complete.

\begin{theorem}
  \label{thm:findqclones-npc}
  Deciding if there exists a $q$-approximate clone of size exactly~$b$ in a
  given election is $\np$-complete.
\end{theorem}

\begin{proof}
  Containment in~$\np$ is clear, since we can verify in polynomial
  time whether a given set of candidates is a $q$-clone in the
  election.
  We present a reduction from the \textsc{Independent Set} problem that, given an undirected graph $G=(U,F)$ and an integer~$k$, asks whether $G$ contains an independent set of size $k$. 
  We let $U=\{u_1,\dots,u_r\}$.

  We  introduce 
  two dummies, $d_1$ and~$d_2$, and for each vertex~$u_i \in U$ we
  create a \emph{vertex candidate} with the same name.   
  We define
  $V=\{v_i \mid i \in [r] \} \cup \{w_f \mid f \in F\}$ as the
  set of voters, with preference orders
  as below (recall the convention about listing sets in
  preference orders): 
  \begin{align*}
    v_{i} &: d_1 \pref d_2 \pref U \setminus \{u_i\} \pref u_i,
    \\
    w_f &: d_2 \pref u_i \pref d_1 \pref U \setminus \{u_i,u_j\} \pref u_j
  \end{align*}
  where $f=\{u_i,u_j\} \in F$ is an edge (we take $i<j$).  We finish
  our construction by setting $b=k$ and $q=r-k$.

  To see the correctness of the reduction, assume first that $Q$ is a
  $q$-clone of size $k$ in our
  election~$E$.  
  Notice that for each vertex candidate $u_i$, both $d_1$ and $d_2$
  are at distance at least~$r+1$ from $u_i$ in some preference order
  (namely, that of $v_i$). Consequently, and since $k>2$, we get that
  $Q \subseteq U$. We claim that the $k$ vertices in~$Q$ form an
  independent set in~$G$.  Indeed, assuming $u_i,u_j \in Q$ are
  connected by an edge of~$G$, we know that $u_i$ and~$u_j$ are at a
  distance~$r+1$ in the preference order of voter~$w_f$, for
  $f=\{u_i,u_j\} \in F$; a contradiction. Conversely, it is
  straightforward to check that an independent set of size~$k$ in~$G$
  forms a $q$-clone in~$E$.
\end{proof}

Fortunately, there are ways to circumvent this hardness result. Below
we show an algorithm that decides if a given election has a size-$b$
$q$-approximate clone in $\fpt(q)$ time.

\begin{theorem}
  \label{thm:findqclones-fpt}
  There is an algorithm that given an election~$E$ and integers $b$ and~$q$
  decides in $\fpt(q)$ time if $E$ contains a size-$b$ $q$-approximate
  clone and, if so, returns one.
\end{theorem}

\begin{proof}
  Let election $E = (C,V)$ and integers $b$ and $q$ be our input.  We
  give an algorithm that either finds a size-$b$
  $q$-approximate clone~$Q$ within~$E$ or concludes that no such
  clone exists. To do so, we employ the bounded search tree
  technique. The algorithm works as follows.

  First, we fix an arbitrary voter~$v \in V$ and guess candidates
  $c', c'' \in C$ such that, supposedly, $v$ ranks $c'$ highest and
  $c''$ lowest among the members of~$Q$. Since $Q$ is to be a size-$b$
  $q$-approximate clone, there must be some $q' \leq q$ such that:
  \[
    \pos_v(c'') - \pos_v(c') + 1 = b+q',
  \]
  that is, $v$ must rank exactly $b+q'$ candidates between~$c'$ and~$c''$ (inclusively). If this is not the case, then our guess
  was incorrect and we terminate the current computation path in the
  search tree without producing a result.  We refer to each candidate
  $d \in C \setminus Q$ such that $v \colon c' \pref d \pref c''$ as
  an \emph{intruder}. We must have exactly $q'$ intruders.  We form a
  set
  \[
    Q_0=\{c \in C \mid c' \pref_v c \pref_v c'' \} \cup \{c',c''\}.
  \]
  Naturally, $Q \subseteq Q_0$, and $Q_0 \setminus Q$ is exactly the set of
  intruders.

  We next build a sequence 
  $Q_0 \supseteq Q_1 \supseteq Q_2 \supseteq \cdots$ of sets where each $Q_i$
  has one intruder fewer than $Q_{i-1}$, and all constructed sets contain~$Q$.
  Given $Q_{i-1}$, we compute $Q_{i}$ as follows: If $|Q_{i-1}| < b$
  then we terminate this computation path as no subset of $Q_{i-1}$
  can be a size-$b$ $q$-approximate-clone.
  Otherwise, we check if 
  there is a voter~$u \in
  V$ 
  for which
  $u$'s most- and least-preferred candidates
  in~$Q_{i-1}$, denoted $a'$ and~$a''$, satisfy
  \[
    \pos_u(a'')-\pos_u(a') +1 > b+q.
  \]  
  If no such voter~$u$ exists, then $Q_{i-1}$ is a
  $q$-approximate clone, and all its size-$b$ subsets are
  $q$-approximate clones as well. Thus, we output one of them arbitrarily,
  and terminate with answer \emph{yes}.  Otherwise, it is clear that
  at least one of~$a'$ and~$a''$ is an intruder that we need to remove. We guess the intruder $\hat{a} \in
  \{a',a''\} \setminus Q$, and set~$Q_i=Q_{i-1} \setminus
  \{\hat{a}\}$.  This finishes the description of the algorithm.
  
  The algorithm terminates at latest after forming the set
  $Q_{q'+1}$ which must have fewer than
  $b$ candidates.  The number of guesses that the algorithm makes
  is therefore at most
  \[
    \underbrace{|C| \cdot (q+1)}_{\text{guesses of $c'$ and $c''$}}
    \cdot \underbrace{\phantom{(}2^q\phantom{)}}_{\substack{\text{guesses of} \\ \text{the intruders}}}.
  \]
  Since each set $Q_{i}$ can be obtained from $Q_{i-1}$ in time
  $\mathcal{O}(|C| \cdot |V|)$, we conclude that the algorithm runs in $\fpt(q)$
  time. The correctness follows from the remarks made in the
  description of the algorithm.
%

\end{proof}

In fact, the algorithm that we use in the above proof
can produce a succinct description of all size-$b$ $q$-approximate
clones that are subsets of a given set $S$ of candidates such that
$|S| \leq q+b$. To do so, whenever the algorithm is about to terminate with a \emph{yes} answer,
it outputs the then-considered set $Q_{i-1}$ indicating that each
of its size-$b$ subsets is a $q$-approximate clone, and proceeds to
the next path within the search tree instead of terminating.

\begin{corollary}\label{cor:all-q-clones}
  There is an algorithm that given an election $E = (C,V)$, integers
  $b$ and $q$, and a subset $S$ of at most $b+q$ candidates from $C$
  outputs in $\fpt(q)$ time a family of sets $S_1, S_2, \ldots$ such
  that a set $Q \subseteq S$ is a size-$b$ $q$-approximate clone if
  and only if it is a subset of some $S_i$.
\end{corollary}

\section{Complexity of Finding Clone Partitions}
\label{sec:partitions}

In this section we discuss the complexity of our \textsc{Clone
  Partition} family of problems for various types of clones. While for
perfect clones the problem is easily seen to be solvable in polynomial
time (indeed, this is essentially a folk result), for imperfect
(i.e., approximate, independent, and subelection) clones the problem
is $\np$-complete. We explore to what extent these intractability
results can be circumvented using parameterized algorithms.

\begin{table}
  \centering
  \begin{tabular}{@{}l@{\hspace{4pt}}c@{\hspace{4pt}}c@{\hspace{4pt}}c@{}}
    \toprule
    & Approximate  & Independent & Subelection\\
    & clones  & clones & clones \\
    \midrule
    $b=2$ & $\p$ \shortciteprop{prop:b=2}           & $\p$ \shortciteprop{prop:b=2}   & $\np$-c \shortcitethm{thm:subelectin_b2_npc}\\
    $b =  3$ & $\np$-c  \shortcitethm{thm:approxclones-npc} & $\np$-c  \shortcitethm{thm:indep_clones_wbn_npc}   & $\np$-c \shortcitethm{thm:subelectin_b2_npc}\\
    \midrule    
    \multirow{2}{*}{param. $t$} & \multirow{2}{*}{$\mathsf{paraNP}$-h \shortcitethm{thm:approxclones-npc}} 
    & $\xp(t)$ \shortciteprop{prop:indep_xp_t} & $\xp(t)$ \shortciteprop{prop:subelection_tractable} \\
        &  & $\mathsf{W}[1]$-h \shortcitethm{thm:indep-clone-part-Whard-t}    & $\mathsf{W}[1]$-h \shortcitethm{thm:subelection_whard} \\
    \midrule
    param. $q$   & $\fpt(q)$ \shortcitethm{thm:approxclone_fpt_q}      & --                 & -- \\
    \midrule
    $w=2$ & --              & $\np$-c \shortcitethm{thm:indep_clones_wbn_npc}    & $\p$ \shortciteprop{prop:subelection_tractable}\\
    \multirow{2}{*}{param. $w$} & \multirow{2}{*}{--}         & \multirow{2}{*}{$\np$-c \shortcitethm{thm:indep_clones_wbn_npc}}    & $\xp(w)$ \shortciteprop{prop:subelection_tractable} \\
     & & & $\mathsf{W}[1]$-h \shortcitethm{thm:subelection_whard} \\
    \midrule
    param. $n$  & $\mathsf{paraNP}$-h \shortcitethm{thm:ApproxClone-NPh-constant-voters}        & $\mathsf{paraNP}$-h \shortcitethm{thm:indep_clones_wbn_npc}                 & $\fpt(n)$ \\
    \bottomrule
  \end{tabular}
  
  \caption{\label{tab:results}Computational complexity of finding partitions into
    imperfect clones. 
    The cells marked with ``--'' refer to parameterizations undefined for the given problem.
    $\np$-c, 
    $\mathsf{W}[1]$-h, and 
    $\mathsf{paraNP}$-h
    stand for $\np$-completeness, 
    $\mathsf{W}[1]$-hardness, and 
    $\paranp$-hardness, respectively.
    }
\end{table}

\begin{proposition}
  \textsc{Perfect-Clone Partition} is in $\p$.
\end{proposition}
\begin{proof}
  Let our input consist of election $E = (C,V)$ over $m$ candidates, together with integers~$b$ and~$t$.
  We fix an arbitrary voter~$v$ and we rename the candidates so that
  \[
    v \colon c_1 \pref c_2 \pref \cdots \pref c_m.
  \]
  For each pair of integers $i, j \in [m]$, $i \leq j$, we test if
  $\{c_i, c_{i+1}, \ldots, c_j\}$ is a clone. We form an instance~$I$
  of \textsc{WEC} with ground set $[m]$ and set family~$\calS$ where
  for each perfect clone $\{c_i, c_{i+1}, \ldots, c_j\}$ the set
  $\{i, i+1, \ldots, j\}$ is put into~$\calS$ with weight $1$. We let
  $t$ be the upper bound on the total weight of a solution for~$I$.
  We solve $I$ using \Cref{pro:wec-q-approx} for $q=0$, and accept if
  and only if $I$ is a yes-instance.

\end{proof}

Before examining our three problems of  partitioning the candidate set into imperfect clones in Sections~\ref{sec:approx}, \ref{sec:indep}, and~\ref{sec:subelection}, we 
show that finding a partition into approximate or independent clones can be easily solved if the maximum allowed size of the clones is $b=2$.

\begin{proposition}
    \label{prop:b=2}
    \textsc{Approximate-} and \textsc{Independent-Clones Partition} are in~$\p$ for $b=2$.
\end{proposition}
\begin{proof}
    Let our input be $E=(C,V)$ with integers $b$, $q$ or~$w$ (depending on the problem), and $t$. 
    We create a graph~$G$ over~$C$ where candidates~$c$ and~$c'$ are connected by an edge if they form a $q$-approximate clone or if they are seen as perfect clones by at least $w$ voters, depending on the problem; note that $G$ can be constructed in polynomial time. Now, we find a maximum-size matching~$M$ in~$G$. Clearly, our instance is a yes-instance exactly if $|M|+(|C|-2|M|)\leq t$.
    
\end{proof}

\subsection{Approximate Clones}
\label{sec:approx}
Focusing now on \textsc{Approximate-Clones Partition}, we present two strong intractability results that show the hardness of this problem even for very restricted instances.

We prove the first of these results,
Theorem~\ref{thm:approxclones-npc}, by a reduction from the following
$\np$-hard problem that we call \textsc{$k$-Partition Into Independent
  Sets}: given an undirected graph $G=(U,F)$, the task is to decide
whether $U$ can be partitioned into $k$ independent sets of
size~$\nicefrac{|U|}{k}$.  This problem is $\np$-hard for $k=3$, since
it is equivalent to the special case of \textsc{3-Coloring} where we
require each color class to be of the same size.  Moreover,
\textsc{$k$-Partition Into Independent Sets} is also $\np$-hard in the
case when $k=\nicefrac{|U|}{3}$, by a reduction from the
\triangulation\ problem~\citep{gar-joh:b:int} whose input is an
undirected graph $H$ and the task is to decide whether we can
partition its vertex set into cliques of size~$3$.  The construction
resembles the one used in the proof of
Theorem~\ref{thm:findqclones-npc} but uses not only two, but a set
of~$k$ dummies.  Setting $k=3$ in our reduction yields hardness for
the setting $t=4$, while setting $k=\nicefrac{|U|}{3}$ yields our
result for $b=3$.

\begin{restatable}[$\star$]{theorem}{thmapproxclonesnpc}
\label{thm:approxclones-npc}
\textsc{Approximate-Clone Partition} is $\np$-complete, even if $b=3$
or $t=4$.
\end{restatable}

The next result uses an intricate reduction from RX3C, relying on \Cref{prop:RX3C-7colorable}.

\begin{restatable}[$\star$]{theorem}{thmApproxCloneNPhconstantvoters}
    \label{thm:ApproxClone-NPh-constant-voters}
    \textsc{Approximate-Clone Partition} is $\np$-hard even if the number of voters is constant ($n=17$).
\end{restatable}


Contrasting Theorems~\ref{thm:approxclones-npc}
and~\ref{thm:ApproxClone-NPh-constant-voters}, we now present an FPT
algorithm for \textsc{Approximate-Clones Partition} with
parameter~$q$, measuring the level of imperfection allowed. Hence, if
$q$ is a small constant, then we can find an approximate-clone
partition efficiently, provided it exists.
\begin{theorem}
\label{thm:approxclone_fpt_q}
  There is an $\fpt(q)$ algorithm for
  \textsc{Approximate}\textsc{-Clone Partition}.
\end{theorem}
\begin{proof}
  Consider an instance $(E,b,q,t)$ of \textsc{Approximate-Clone Partition}
  with election $E = (C,V)$. 
  Our main idea is to convert it into an instance of the \textsc{WEC},
  where the sets are intervals over $[m]$, for $m=|C|$, with possible
  ``holes'' on the first and last $2q$ entries. Such instances of
  \textsc{WEC} can be solved in $\fpt(q)$ time using
  Proposition~\ref{pro:wec-q-approx}.

  \paragraph{Initializing the WEC Instance.}
  Let $m = |C|$ be the number of candidates in $E$ and let us fix an
  arbitrary voter~$v \in V$. We rename the candidates so that
  $C = \{c_1, \ldots, c_m\}$ and
  \[
    v: c_1 \pref_v c_2 \pref_v \cdots \pref_v c_m.
  \]
  We form an instance of \textsc{WEC} with underlying set
  $[m]$, where each $i \in [m]$ corresponds to  candidate
  $c_i \in C$. The sets included in the family $\calS$ for
  \textsc{WEC} will correspond to (segments of) $(q,b)$-approximate
  clones, defined below.

  \paragraph{Segments.}
  Let $Q$ be some $(q,b)$-approximate clone and let $a$ and $b$ be its
  top- and bottom-ranked members according to~$v$. We say that a set
  $P$ of candidates is \emph{enclosed by}~$Q$ if for each $c \in P$
  we have 
  $a \pref_v c \pref_v b$.
  A set~$S$ of candidates is a
  \emph{segment} if it can be partitioned into a $(q,b)$-approximate
  clone~$Q$ and a collection~$P_1,\dots,P_k$ of additional
  $(q,b)$-approximate clones such that each $P_i$ is enclosed by $Q$.
  We refer to $Q, P_1, \ldots, P_k$ as an \emph{implementation of $S$}, with
  $Q$ being its \emph{base}, and we let $1+k$ be its \emph{size}. The size of a
  segment $S$, denoted $\size(S)$, is the smallest among the sizes of
  its implementations.

  Given a segment $S$, we define its \emph{signature} as a length-$m$ binary
  string such that for each $i \in [m]$, its $i$-th character is~$1$
  if $c_i$ belongs to $S$ and it is $0$ otherwise. Note that since a
  segment's base is a $(q,b)$-approximate clone, the segment's
  signature can have at most $q$ $0$s between its leftmost and
  rightmost~$1$.

  An important observation is that an instance of
  \textsc{Approximate-Clone Partition} is a yes-instance if and only
  if there is a partition of the candidate set $C$ into segments
  $S_1, \ldots, S_{t'}$ such that none of them is enclosed by any
  other and $\sum_{i=1}^{t'} \size(S_i) \leq t$. We observe that each
  $S_i$ in such a partition has the following property: If a $0$
  appears in its signature between the leftmost and rightmost $1$,
  then it must be within $2q$ positions either from the leftmost $1$
  or from the rightmost $1$. Indeed, otherwise one of the segments
  would be enclosed by another one (or would not be a segment at all,
  as each its implementation would require the base not to be a
  $(q,b)$-approximate clone). We refer to the segments that satisfy
  this property as \emph{valid}. Similarly, we say that signatures of
  valid segments are valid.

  \paragraph{Adding Sets to the WEC Instance.} For each valid segment
  $S$, we form a set for our \textsc{WEC} instance whose weight is
  $\size(S)$ and that includes exactly those elements $i \in [m]$ for
  which $c_i$ belongs to $S$. We fix the threshold on the total weight of the 
  desired solution for our \textsc{WEC} instance to be $t$.  Note that there
  are at most $\mathcal{O}({4q \choose q} m^2) \leq \mathcal{O}(2^{4q} m^2)$ (signatures
  of) valid segments: we have $m^2$ choices for the positions of the
  left- and rightmost $1$ in the signature, and between them we have
  at most~$q$ $0$s, each located at most $2q$ positions from the left-
  or the rightmost $1$.  Hence, it suffices to show that we can
  compute all valid segments.

  \paragraph{Computing Valid Segments and Their Sizes.} Given a valid
  signature, we show how to verify in $\fpt(q)$ time if it indeed
  corresponds to a segment and compute its size. 
  
  First, we let $S = \{e_1, \ldots, e_s\}$ be the set of
  candidates that correspond to the $1$s in the signature (so if the
  first $1$ is on position $i$ in the preference order of~$v$, then $e_1$ is $c_i$, and so on). Thus, $S$ is the candidate set about which we want to decide whether it is a segment and, if so, establish its size. Suppose
  that $Q, P_1, \ldots, P_k$ is an implementation of~$S$ with the
  smallest size. We guess the size~$b'$ of~$Q$; we have $b' \leq b$.
  
  Next, we apply \Cref{cor:all-q-clones} for the candidate set~$S$
  (recall that $|S| \leq b+q$) and size~$b'$, and guess a
  set~$Q' \subseteq S$ from its outputs, such that $Q$ is a size-$b'$
  subset of~$Q'$ and, moreover, every $b'$-sized subset of~$Q'$ is a
  $(q,b')$-approximate clone.
  
  We proceed to search for
  $P_1, \ldots, P_k$ that together with
  $Q = Q' \setminus (P_1 \cap \cdots \cap P_k)$ have the following
  properties:
  \begin{enumerate}
  \item Each of $P_1, \ldots, P_k$ is a $(q,\min(q,b))$-approximate
    clone; indeed we are looking for a $(q,b)$-approximate clone
    partition, so each $P_i$ needs to be $(q,b)$-approximate, and at
    most $q$ members of $S$ can belong to $P_1, \ldots, P_k$, so each
    of them must be $(q,q)$-approximate.
  \item $P_1\cup  \ldots \cup P_k$ includes all the candidates in
    $S \setminus Q'$ as the segment~$S$ must include them, but they are not in $Q$.
  \item $|Q| = b'$. 
  \end{enumerate}
  
  Due to the first property, for each $i \in [k]$ there is
  a $j$ such that $P_i$ is a subset of
  $\{e_j, e_{j+1}, \ldots, e_{j+2q}\}$ of size at most $q$.  Since for
  each such candidate set we can decide whether it forms a
  $(q,\min(q,b))$-approximate clone, it follows that we can compute in
  $\fpt(q)$ time a sequence $\calR=R_1, R_2, \ldots$ of
  $(q,\min(q,b))$-approximate clones such that each $P_i$ is equal to
  some $R_\ell$ from this list~$\calR$. 
  
  Furthermore, for each
  candidate $e \in S \setminus Q'$ there are at most $2^{2q}$ clones
  in the list~$\calR$ that include $e$. Hence, for each candidate
  $e \in S \setminus Q'$ we guess the clone $R_\ell \in \calR$ that
  includes~$e$ and belongs to the implementation~$Q,P_1,\dots,P_k$
  of~$S$, ensuring that each of these guessed clones is either
  disjoint from the other ones or equal to one already guessed (as the
  same clone~$R_\ell$ may include more than one candidate from
  $S \setminus Q'$). This way we obtain some first
  $P_1, \ldots, P_{k'}$ clones that satisfy the second of the above
  conditions, i.e., each $e \in S \setminus Q'$ is contained in
  some~$P_j$, for $j \in [k']$.

  Now we move on to compute the remaining $k-k'$ clones
  $P_{k'+1}, \ldots, P_{k}$, so that
  $|Q' \setminus (P_1 \cup \cdots \cup P_k)| = b'$. We let
  $Q'' = Q' \setminus (P_1 \cup \cdots \cup P_{k'})$ and $b'' =
  |Q''|$. We need $P_{k'+1}, \ldots, P_k$ to be disjoint subsets of
  $Q''$, each chosen from~$\calR$, whose total cardinality is
  exactly $b''-b' \leq q$. Finding such a subset in $\fpt(q)$ time is
  possible by invoking the \textsc{Partial Set Cover} algorithm of
  \citet{bla:j:fpt-partial-cover}. His algorithm---based on
  color coding---finds a minimum-weight family of sets that covers at least a given number
  of elements; however, by only very minor adaptations of the dynamic programming part of his algorithm we can guarantee that the returned  set family (if there is one) contains pairwise disjoint sets whose total size is exactly $b''-b'$. 

  Altogether, after considering all the guesses, we either find that our
  signature corresponds to a segment and output its size (the smallest
  size of an implementation that we have found), or none of the guesses
  led to finding an implementation and the signature does not
  correspond to a segment.


  \paragraph{Putting It All Together.}
  Let us summarize our algorithm.
  \begin{enumerate}
  \item We consider all valid signatures, for each of them checking if
    it corresponds to a segment and computing its size. We collect all
    such segments 
    in a \textsc{WEC} instance, where the segments' weights are their
    sizes.
  \item We solve the \textsc{WEC} instance by applying
    \Cref{pro:wec-q-approx}, by observing that in the language of
    \textsc{WEC}, our segments are $2q$-approximate intervals.
  \item We accept if the answer for our \textsc{WEC} instance is yes.
  \end{enumerate}
  The correctness and $\fpt(q)$ running time follow from the preceding
  discussion.

\end{proof}

\subsection{Independent Clones}
\label{sec:indep}

Regarding independent-clone partitions, our first result rules out an
efficient algorithm for finding such a partition even in the case when
we have a constant number of voters and we are searching for clones of
size at most three. It follows by a reduction from \textsc{RX3C}
and relies on \Cref{prop:RX3C-7colorable}.

\begin{restatable}[$\star$]{theorem}{thmindepcloneswbnnpc}
\label{thm:indep_clones_wbn_npc}
  \textsc{Independent-Clone Partition} is $\np$-complete, even if we
  require clones of size at most $b = 3$, each recognized by at least
  $w= 2$ voters, and the number of voters is a constant ($n=14$).
\end{restatable}

Yet, if the number of clones in the desired partition (that is, $t$) is
constant, we can find this
partition 
in polynomial-time by a simple brute-force approach.

\begin{proposition}
\label{prop:indep_xp_t}
  \textsc{Independent-Clone Partition} has an $\xp$ algorithm with parameter~$t$,  the number of clones.
\end{proposition}
\begin{proof}
  We take all the preference orders of the voters in the input  election~$(C,V)$ and merge them into
  a single list. Then, in this list, we guess at most $t$ intervals;
  there are $\mathcal{O}((mn)^{2t})$ guesses to consider, where ${m=|C|}$ and~${n=|V|}$. For each group of
  intervals, we verify in polynomial time if it indeed corresponds to an
  independent-clone partition that meets the conditions of the
  input instance.
\end{proof}


Next, we show that \textsc{Independent-Clones Partition} is
$\mathsf{W}[1]$-hard for parameter $t$, which means that under
standard complexity-theoretic assumptions we cannot improve the above
algorithm to be an $\fpt$ one. We prove this 
by giving a parameterized reduction from the \textsc{Perfect Code}
problem whose input is an undirected graph~$G$ and an integer~$k$, and
the task is to decide whether $G$ admits a set~$S$ of~$k$ vertices
such that each vertex of~$G$ is contained in the closed neighborhood
of exactly one vertex from~$S$. In fact, we need the multicolored
version of this problem---where each vertex of~$G$ is assigned a color
from~$[k]$ and we require $S$ to contain one vertex from each color
class. $\np$-hardness of this variant follows from the $\np$-hardness
of \textsc{Perfect Code}~\cite{dow-fel:j:fixed-parameter-I} by
standard techniques.

\begin{restatable}[$\star$]{theorem}{thmindepclonepartWhardt}
    \label{thm:indep-clone-part-Whard-t}
    \textsc{Independent-Clones Partition} is \Wh{1} with parameter~$t$, even if $w=2$.
\end{restatable}

\subsection{Subelection Clones}
\label{sec:subelection}

While for approximate and independent clones we found efficient
algorithms for computing partitions into clones of size at most two
(recall \Cref{prop:b=2}), for subelection clones the same task is
intractable.
%
\begin{restatable}[$\star$]{theorem}{thmsubelectinbnpc}
\label{thm:subelectin_b2_npc}
  \textsc{Subelection-Clone Partition} is $\np$-complete, even if the
  size of the clones is upper-bounded by~$2$ (i.e., $b=2$).
\end{restatable}

\iftrue
\begin{proof}
  To see that the problem is in $\np$, it suffices to guess a
  subelection and its perfect-clone partition, and verify that they
  meet the requirements of the given instance.

  Next, we give a reduction from \textsc{RX3C}.  Let $(X,\calS)$
  be our input instance where $X = \{x_1, \ldots, x_{3k}\}$ is a set
  of $3k$ elements and $\calS = \{S_1, \ldots, S_{3k}\}$ is a family
  of size-$3$ subsets of $X$. We form an election $E = (C,V)$ where
  \[
    C = \{ s_1, s^+_1, s^-_1, \ldots, s_{3k}, s^+_{3k}, s^-_{3k}\}
  \]
  and $V$ contains the following sets of voters:
  \begin{enumerate}
  \item There are $9k^2$ \emph{type-1 voters}, each with preference order
    \[
       s^+_1 \pref s_1 \pref s^-_1 \pref s^+_2 \pref s_2 \pref s^-_2 \pref \cdots
    \]
    and $9k^2$ \emph{type-2 voters}, each with
    preference order
    \[
       s^-_1 \pref s_1 \pref s^+_1 \pref s^-_2 \pref s_2 \pref s^+_2 \pref \cdots
    \]

  \item For each element $x_i \in X$ and each set $S_j$ that contains~$x_i$, we form a \emph{coverage voter} $v_{i,j}$ whose preference
    order is the same as that of type-1 voters, except that we have the following modifications: 
    \begin{enumerate}
        \item $v_{i,j} \colon s^-_j \pref s^+_j \pref s_j$, and 
        \item for each
    $\ell \in [3k]$, if
    $S_\ell$ contains $x_i$ but $\ell \neq j$, then
    $v_{i,j} \colon s^+_\ell \pref s^-_\ell \pref s_\ell$.
    \end{enumerate} 
      There are at most
    $9k^2$ coverage voters.

  \end{enumerate}
  We set the upper bound on the clone sizes to be $b=2$, the maximum
  number of clones to be $t = {6k}$, and the number of voters in the
  subelection to be $w = 18k^2+3k$.

  Let us assume that there is a subelection $E' = (C,V')$ with at
  least $w = 18k^2 + 3k$ voters that has a perfect-clone partition
  $\Pi=(C_1, \ldots, C_{t'})$ with $t' \leq t = 6k$ and $|C_i| \leq 2$ for each
  $i \in [t']$.  By a simple counting argument,
  $E'$ must include some type-1 and some type-2 voters and, hence, we
  can assume it includes all of them. Due to the preference orders of
  type-1 and type-2 voters, for each set~$S_j$, either 
  (i)~$\Pi$ contains  $\{s_j, s_j^+\}$ and $\{s^-_j\}$, or 
  (ii)~$\Pi$ contains $\{s_j, s_j^-\}$ and
  $\{s^+_j\}$, or 
  (iii)~$\Pi$ contains each of $\{s_j\}, \{s^+_j\}, \{s^-_j\}$. 
  However, this last
  option is impossible, as $|\Pi| \leq 6k$.
  Intuitively, if 
   $\{s_j, s^+_j\} \in \Pi$ then $S_j$ is included in a solution
  for the \textsc{RX3C} instance, and if 
  $\{s_j, s^-_j\} \in \Pi$ then it is not.

  Next, we observe that if $E'$ includes some voter $v_{i,j}$ then,
  due to $v_{i,j}$'s preference order and the preceding paragraph, we must
  have $\{s_j, s^+_j\} \in \Pi$ and for every
  $\ell \in [3k] \setminus \{j\}$ such that $x_i \in S_\ell$ we must
  have $\{s_\ell, s^-_\ell\} \in \Pi$. This means that for each
  $i \in [3k]$, $E'$ contains at most one coverage voter for
  $x_i$. However, since $E'$ contains $18k^2+3k$ voters, it must
  contain one coverage voter for each $x_i \in X$. This means that
  the sets $S_\ell$ for which 
  $\{s_\ell, s^+_\ell\} \in \Pi$ must
  form an exact cover of $X$.

  For the other direction, assume that $(X,\calS)$ is a yes-instance of the \textsc{RX3C} problem and that
  there are $k$ sets from $\calS$ whose union is $X$. For simplicity,
  let us assume that these are $S_1, \ldots, S_k$. We form a clone
  partition that contains clones $\{s_i, s^+_i\}$
  and $\{s^-_i\}$ for each $i \in [k]$, and clones $\{s_\ell, s^-_\ell\}$ and $\{s^+_\ell\}$ for each $\ell \in [3k] \setminus [k]$. We form a subelection
  $E'$ that includes all type-1 voters, all type-2 voters, and for
  each $x_i$ we include a single coverage voter $v_{i,j}$ such that
  $j \in [k]$ and $x_i \in S_j$ (this is possible since
  $S_1, \ldots, S_k$ form an exact cover of $X$). This shows that we
  have a yes-instance of our \textsc{Subelection-Clone Partition}
  instance.
  
\end{proof}
\fi 

On the positive side, we do have an $\fpt$ algorithm parameterized by
the number~$n$ of voters. This strongly contrasts our results 
for the other type of
clones, where this parameterization leads to intractability. We
also find $\xp$ algorithms for the parameterizations by the size $t$
of the partition and the number $w$ of voters needed to recognize
the clones.

\begin{proposition}
\label{prop:subelection_tractable}
  \textsc{Subelection-Clone Partition} has an $\fpt$ algorithm for the
  parameterization by the number~$n$ of voters, an $\xp$ algorithm for
  parameterization by the number~$w$ of voters  in the subelection, and an $\xp$
  algorithm for the parameterization by the number~$t$ of clones.
\end{proposition}
\begin{proof}
  The $\fpt(n)$ algorithm guesses the voters that are included in the
  subelection and then runs the polynomial-time algorithm for
  \textsc{Perfect-Clone Partition}. The $\xp(w)$ algorithm proceeds
  similarly, by explicitly guessing $w$ voters.  The $\xp(t)$
  algorithm first guesses a single voter to be included in the
  subelection, and then a partition of its preference order into at
  most $t$ clones of appropriate sizes (i.e., it guesses a
  size-at-most-$t$ perfect-clone partition consistent with this vote,
  where each clone contains at most $b$ candidates). Finally, it
  checks if there are at least $w-1$ additional voters that recognize
  all of these guessed clones. 
  
\end{proof}

\begin{corollary}
  \textsc{Subelection-Clone Partition} can be solved in polynomial
  time for parameter $w = 2$.
\end{corollary}

The $\mathsf{XP}$ algorithms from \Cref{prop:subelection_tractable}
cannot be improved to $\fpt$ ones as the problem is
W[1]-hard 
even for $w+t$. The reduction is from the classic \textsc{Multicolored
  Clique}
problem~\cite{fel-her-ros-via:j:multicolored-hardness,pie:j:mulitcolored-clique}.


\begin{restatable}[$\star$]{theorem}{thmsubelectionwhard}
\label{thm:subelection_whard}
    \textsc{Subelection-Clone Partition} is \Wh{1} with parameter~$w+t$.
\end{restatable}

\section{Summary, Conclusions, and Future Work}
We have introduced three natural types of imperfect clones and we have
studied the complexity of (a)~identifying them in elections and of
(b)~partitioning the candidates into clones of these types, with desired
properties (such as the number of clones, their sizes, or their levels
of imperfection). Overall, our problems are largely intractable, but
we also found some polynomial-time, $\fpt$, and $\xp$ algorithms for
special cases. The main conclusion from our work is that if one were
to solve our problems in practical settings, most likely one should
use heuristic approaches, such as ILP solvers.

Our work can be extended in several ways. Foremost, while we tried to
separate the different types of imperfect clones, it would be more
realistic to study clones that can be both approximate and independent
at the same time (i.e., only some voters would be required to
recognize them, and even they could do so approximately).
It would also be very interesting to perform experimental analysis of
imperfect clones in elections. It is interesting to see how often they
indeed appear, and for what levels of imperfection.
Next,
\citet{cor-gal-spa:c:sp-width,cor-gal-spa:c:spsc-width} described how
one can use clones to generalize classic
single-peaked~\citep{bla:b:polsci:committees-elections} and
single-crossing~\citep{mir:j:single-crossing,rob:j:tax} domains, while
maintaining good computational properties of resulting elections. It
might be interesting to analyze if imperfect clones can be used in a
similar way.
Finally, it would be natural to consider approximate
clones not only in the ordinal setting, but also in the approval
one. In particular, this would open up the possibility of clone
analysis in participatory budgeting elections, which often use approval
data~\citep{rey-mal:t:pb-survey}.

\newpage

\section*{Acknowledgements}
This project has received funding from the European Research Council
(ERC) under the European Union’s Horizon 2020 research and innovation
programme (grant agreement No 101002854).
Ildikó Schlotter was further supported by the Hungarian Academy of Sciences under its Momentum
Programme (\mbox{LP2021-2}) and its János Bolyai Research Scholarship. Grzegorz Lisowski acknowledges support by the European Union under the Horizon Europe project Perycles (Participatory
Democracy that Scales). 

\begin{center}
  \includegraphics[width=3cm]{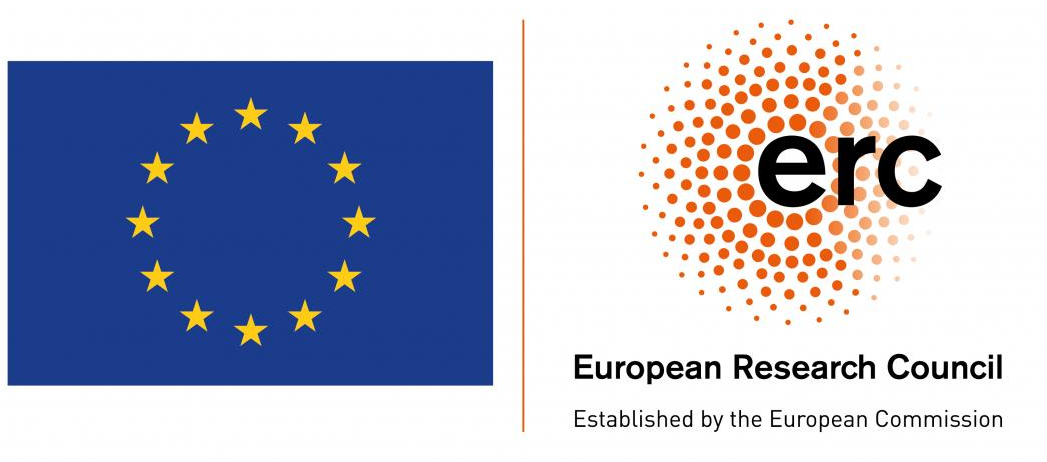}
\end{center}

\bibliographystyle{ACM-Reference-Format} 
\bibliography{bib}

\ifshort

\else
\clearpage
\begin{appendices}
\section{Missing proofs from Section~\ref{sec:approx}}

\thmapproxclonesnpc*

\begin{proof}
  Containment in $\np$ is clear, as we can check in polynomial time
  whether a given partitioning of the candidate set indeed fulfills
  the necessary conditions.
    
  We present a reduction from the following $\np$-hard problem that we
  call \textsc{$k$-Partition Into Independent Sets}: given an
  undirected graph $G=(U,F)$, the task is to decide whether $U$ can be
  partitioned into $k$ independent sets of size~$\nicefrac{|U|}{k}$.
  This problem is $\np$-hard for $k=3$, since it is equivalent with
  the special case of \textsc{3-Coloring} where we require each color
  class to be of the same size; note that by adding sufficiently many
  isolated vertices, we can ensure that 3-colorability is equivalent
  with having 3 independent sets of the same size.  Moreover,
  \textsc{$k$-Partition Into Independent Sets} is also $\np$-hard in
  the case when $k=\nicefrac{|U|}{3}$, by a reduction from the
  \triangulation\ problem~\citep{gar-joh:b:int} whose input is an
  undirected graph $H$ and the task is to decide whether we can
  partition its vertex set into cliques of size~$3$: a triangulation
  exists for~$H$ exactly if its complement can be partitioned into
  independent sets of size~$3$.  Let $U=\{u_1,\dots,u_r\}$; we may
  assume w.l.o.g.\ that $r$ is divisible by~$k$.

  We construct an instance of \approxclones\ over an ordinal
  election~$E=(C,V)$ in a manner resembling our construction in Theorem~\ref{thm:findqclones-npc}.  Let us set $b=\nicefrac{|U|}{k}$. 
  
  For  each vertex~$u_i \in U$ we create a \emph{vertex candidate}
  with the same name, and we introduce a set~$D=\{d_1,\dots,d_b\}$ of
  dummies, so the set of candidates is $C=U \cup D$.
    
  We define
  $V=\{v_{i,j}:i \in [b], j \in [r] \} \cup \{w_f: f \in F\}$ as the
  set of voters whose preferences
  are defined
  as 
  \begin{align*}
    v_{i,j} &: d_i \pref D \setminus \{d_i\} \pref U \setminus \{u_j\} \pref u_j,
    \\
    w_f &: D \setminus \{d_1\} \pref u_i \pref d_1 \pref U \setminus \{u_i,u_j\} \pref u_j
  \end{align*}
  where $f=\{u_i,u_j\}$ for some $i<j$.  We finish our construction by
  setting $b=\nicefrac{|U|}{k}$, $q=r-b$, and $t=k+1$.

  \smallskip 
  We claim that the constructed instance of \textsc{Approximate-Clones Partition} is a yes-instance if and only if $(G,k)$ is
  a yes-instance of \textsc{$k$-Partition Into Independent Sets}.%
  \smallskip

  Assume first that there exists a partitioning~$\Pi$ of~$C$ into at
  most~$t$ $q$-approximate clones, each of size at
  most~$b$. 
  Observe that no dummy candidate~$d_i \in D$ can be contained in the
  same $q$-approximate clone in~$\Pi$ as some vertex candidate
  $u_j \in U$, because the distance of $d_i$ and~$u_j$ within the
  preference list of~$v_{i,j}$ is more than~$q+b=r$. By $|D|=b$, we
  need at least one clone in~$\Pi$ that contains only dummies, and by
  $|U|=kb$ we need at least~$k$ clones in~$\Pi$ that contain only
  vertex candidates. Hence, $t=k+1$ implies that $D \in \Pi$ and
  that $\Pi$ contains $k$ partitions that each consist of~$b=k$
  vertex candidates. We claim that each such candidate set is an
  independent set in~$G$. Indeed, if $u_i$ and~$u_j$ are contained in
  the same $q$-approximate clone~$Q \in \Pi$ but
  $\{u_i,u_j\}=f \in F$, then $u_i$ and~$u_j$ are at a distance of
  $r+1>b+q$ in the preference list of voter~$w_f$; a
  contradiction. Hence, $\Pi$ yields a partitioning of~$U$ into~$k$
  independent sets of size~$b=\nicefrac{|U|}{k}$.
 
  For the other direction, assume that $U_1,\dots,U_k$ partition the
  vertex set~$U$ of~$G$ into independent sets of
  size~$\nicefrac{|U|}{k}=b$. Then $\Pi=\{D\} \cup \{U_i:i \in [k]\}$
  is a partitioning of~$C$ into exactly $t$ $q$-approximate clones,
  each of size exactly~$b$. To see this, it suffices to observe each
  pair $(u_i,u_j)$ of candidates where $\{u_i,u_j\} \notin F$ is
  within a distance of at most~$q+b=r$ in each voters' preferences:
  this is clear for all voters of the form~$v_{i',j'}$, and it also
  holds for each voter $w_f$, $f \in F$, because due to
  $\{u_i,u_j\} \notin F$ it is not possible that
  candidates~$u_i$ and~$u_j$ are the first- and last-ranked candidates
  within the vote of~$w_f$ among the candidates of~$U$. This proves
  our claim and, hence, the correctness of the reduction.  Note that
  the $\np$-hardness result for $b=3$ follows from setting
  $k=\nicefrac{|U|}{3}$, and for $t=4$ from setting $k=3$.
\end{proof}

\thmApproxCloneNPhconstantvoters*
\begin{proof}
    We have already shown containment in $\np$ in Theorem~\ref{thm:approxclones-npc}. 
    We show $\np$-hardness by a reduction from RX3C. 
    Let $(X,\calS)$ with $X=\{x_1,\dots,x_{3k}\}$ be an input instance for RX3C, and recall that $|\calS|=|X|$. 
    Let $\calS_1,\dots,\calS_7$ be the partitioning of~$\calS$ guaranteed by Proposition~\ref{prop:RX3C-7colorable}. 
    Let us write $\calS=\{S_i^1,\dots,S_i^{\nu_i}\}$ for each $i \in [7]$. 

    \paragraph{Construction.}
    We set $b=4k$ and $q=b^2-b$, and define an integer $p$ so that $p+1$ is the smallest prime larger than $(b+2)^2$.

     For each ${x_i \in X}$, we create a corresponding \emph{item candidate}, also denoted by~$x_i$. 
    We further create a \emph{triple candidate}~$t(S)$ and a set $A(S)$ of $b-4$ \emph{auxiliary candidates} for each set $S \in \calS$. 
    Additionally, we create \emph{dummy sets} $D_0,D_1,\dots,D_p$ with $|D_0|=b-2k=\frac{b}{2}$ and $|D_i|=b$ for each $i \in [p]$. 
    The candidate set is thus 
    \[C=X \cup A \cup T \cup \bigg(\bigcup_{i \in [p] \cup \{0\}} D_i \bigg)\]
    where $A=\bigcup_{S \in \calS} A(S)$ and $T=\{t(S):S \in \calS\}$. 
    
    We first create two voters, $v$ and $v'$, whose preferences are
    \begin{align*}
        v&: X \cup A \cup T \pref D_0 \pref D_1 \pref \dots \pref D_p, \\
        v' &: X \cup A \cup T \pref D_0 \pref D_{\varphi(1)} \pref \dots \pref D_{\varphi(p)}
    \end{align*}
    where ${\varphi(i)}=i(b+2) \bmod (p+1)$; since $p+1$ is a prime, the above preferences are well defined.
    Next, we create a voter~$u$ with 
    \begin{align*}
        u &: X \cup A \pref D_p \pref
        \underbrace{T \pref \dummyseries_{h} \pref D_0}_{b+q} \pref [\cdots]
    \end{align*}
    where $h=b+q-|T \cup D_0|$ and we write $\dummyseries_\ell$ for a series of $\ell$ dummies for each integer~$h$, filled up from left to right with candidates from $D_1, D_2$, and so on.
    Note that $|T \cup D_0|=b+k<b+q$, so $u$'s preferences are also well defined.
    
    Next, for each $i \in [7]$ we create two voters, $w_i$ and~$w'_i$. 
    To define their preferences, we need some further notation. Let 
    \begin{align*}
        Z(S)&= S \cup A(S) \qquad \textrm{for each } S \in \calS; \\
    Z^{-i}& = (X \cup A) \setminus  \bigg( \bigcup_{S \in \calS_i} Z(S)\bigg); \\
    T^{-i}&=T \setminus \{t(S): S \in \calS_i\}.
    \end{align*}
    Note that $|Z(S)|=b-1$ for each $S \in \calS$, and that $Z(S)$ and $Z(S')$ are disjoint for each $S,S' \in \calS_i$, for all $i \in [7]$.
    Using the same notation as before, the preferences of $w_i$ and~$w'_i$ are 
    \begin{align*}
        w_i &: Z^{-i} \pref \overbrace{\underbrace{Z(S_i^1)}_{b-1} \pref \cdots \pref \underbrace{Z(S_i^{\nu_i})}_{b-1} 
        \pref T^{-i}
        \pref \dummyseries_{h_i} }^{q+b-1}   \\
        & \pref \underbrace{t(S_i^1) \pref \dummyseries_{b-2}}_{b-1} \pref \dots 
        \pref \underbrace{t(S_i^{\nu_i-1}) \pref \dummyseries_{b-2}}_{b-1} \\
        & \pref t(S_i^{\nu_i}) \pref D_0 \pref [\cdots], \\ 
        w'_i &: Z^{-i} \pref \overbrace{\underbrace{Z(S_i^{\nu_i})}_{b-1} \pref \cdots \pref \underbrace{Z(S_i^1)}_{b-1} 
         \pref T^{-i}
        \pref \dummyseries_{h_i}}^{q+b-1}   \\
        & \pref \underbrace{t(S_i^{\nu_i}) \pref \dummyseries_{b-2}}_{b-1} \pref \dots 
        \pref \underbrace{t(S_i^2) \pref \dummyseries_{b-2}}_{b-1} \\
        & \pref t(S_i^1) \pref D_0 \pref [\cdots]
    \end{align*}
    where $h_i=b+q-1-\nu_i(b-1)-|T^{-i}| >0$.
    
    To finish our instance of \textsc{Approximate-Clone Partition}, we set $t=\nicefrac{|C|}{b}$. 
    
    \paragraph{Correctness.}
    Assume first  that $\Pi$ is a partitioning of~$C$ into at most~$t$ $(q,b)$-approximate clones; we know $|\Pi|=t$ because $|C|=tb$.

    By $p+1>(b+2)^2$, the sets $\{\varphi(i-b-1),\dots,\varphi(i+b+1)\}$ and 
    $\{i-b-1, \dots, i+b+1\}$ are disjoint for each integer $i \in [p] \cup \{0\}$ (where integers in the latter set are interpreted modulo $p+1$). Hence, for each two dummy set~$D_i$ and $D_j$, there are at least $b+1$ dummy sets (of total size at least $b^2+\frac{b}{2}$) strictly ranked between $D_i$ and~$D_j$ either in the preferences of~$v$ or in the preferences of~$v'$. Due to $b+q=b^2$ this implies that no two dummies in different dummy sets can be in the same partition of~$\Pi$. 
    Furthermore, the preferences of~$v$ and~$v'$ imply that the only dummies that can be contained in the same partition of~$\Pi$ as some candidate in $X \cup A \cup T$ are dummies in~$D_0$.
    Considering now the preferences of voter~$u$, we see that dummies in~$D_0$ can never be in the same partition of~$\Pi$ with any candidate in~$X \cup A$. 
    
    Summing this up, we deduce that any partition $\pi \in \Pi$ that contains some dummy from~$D_0$ must satisfy $\pi \subseteq T \cup D_0$. By $|D_0|=2k$, $|T|=3k$, and $b=4k$, we obtain that there exists exactly one clone in~$\Pi$ containing candidates from~$D_0$, and it must have the form $D_0 \cup T_0$ for a set $T_0 \subset T$ of $2k$ triple candidates. 
    
    Consider now the partition $\pi(t) \in \Pi$ containing candidate $t=t(S_i^j) \in T \setminus T_0$ for some $i \in [7]$ and $j \in [\nu_i]$.
    Looking at the preferences of~$w_i$, we can observe that the only candidates from~$A \cup X$ not farther than~$b+q$ from~$t$ within the preferences of~$w_i$ are the candidates in $Z(S_i^h)$ for $j \leq h \leq \nu_i$.  Similarly, the only candidates from~$A \cup X$ not farther than~$b+q$ from~$t$ within the preferences of~$w'_i$ are the candidates in $Z(S_i^h)$ for $1 \leq h \leq j$.  
    This means that \[\pi(t(S_i^j)) \cap (A \cup X) \subseteq Z(S_i^j).\] 
    Since this holds for each $t(S_i^j) \in T \setminus T_0$, we know that no two triple candidates $t,t' \in T \setminus T_0$ can be placed in the same partition of~$\Pi$, as such a partition could only contain the $k$ triple candidates $T \setminus T_0$ but no candidates from~$A \cup X$. Thus, for each $t(S_i^j) \in T \setminus T_0$ we must have $\pi(t(S_i^j))=Z(S_i^j) \cup \{t(S_i^j)\}$ due to $|\pi(t(S_i^j))|=b$. Therefore, the $k$ sets $S_i^j$ where $t(S_i^j) \in T \setminus T_0$ give a partitioning of the item candidates, i.e., of~$X$, into disjoint 3-sets from $\calS$, proving that our RX3C instance is a yes-instance.

    Assume now that $\calS' \subseteq \calS$ is a collection of $k$ sets whose union is~$X$.
    Consider the partitioning~$\Pi$ that consists of 
    \begin{itemize}
        \item the set $S \cup A(S) \cup \{t(S)\}$ for each $ S \in \calS'$,
        \item an arbitrary partitioning of $A^-=A \setminus (\bigcup_{S \in \calS'} A(S))$ into subsets of size~$b$ (note that $|A^-|=2k(b-4)$ is divisible by $b=4k$), 
        \item the set $D_0 \cup T \setminus T'$ where $T'=\{t(S):S \in \calS'\}$, and
        \item all dummy sets $D_i$, $i \in [p]$.
    \end{itemize}
    
    It is straightforward to verify that all of the above sets have size~$b$ and are $q$-approximate clones; let us give the key ideas. 
    First, note that \[|X \cup A \cup T|+b=3k(b-2)+b<b^2=b+q,\] so candidates in $X \cup A \cup T$ are at a distance at most $b+q$ from each other in the votes of~$v,v'$, and~$u$. By the discussion above, we also know that for each $S \in \calS$, candidates in~$Z(S)$ are within a distance of $q+b$ from~$t(S)$ within the preferences of~$w_i$ and~$w'_i$ irrespective of whether $S \in \calS_i$ or not (for the latter, observe that all candidates in~$T^{-i}$ are within distance $b+q$ from all candidates in~$Z^{-i}$). Additionally, candidates from~$D_0 \cup T$ are within a distance at most~$b+q$ from each other in every vote, and the same holds for each dummy set $D_i$, $i \in [p]$, as well as for the set~$A$. 
    Hence, $\Pi$ is a partitioning fulfilling all the desired conditions. 
    
\end{proof}

\section{Missing proofs from Section~\ref{sec:indep}}
\thmindepcloneswbnnpc*
\begin{proof}
  To see that the problem is in $\np$, note that in order to verify that an instance $(E,b,q,t)$ of the problem is a yes-instance it suffices to guess a
  partition of the candidate set into clones, and then verify that there are at
  most $t$ of them, all of them have size at most~$b$, and candidates in each are ranked consecutively by at least $w$ voters.

  Next, we show that the problem is $\np$-hard by giving a reduction
  from \textsc{RX3C}. 
  Let $X=\{x_1,\dots,_{3k}\}$ and $\calS$ be our input instance for \textsc{RX3C}.
  Let $\calS_1,\dots,\calS_7$ be the partitioning obtained by Proposition~\ref{prop:RX3C-7colorable}, let $\calS_i=\{S_i^1,\dots,S_i^{\nu_i}\} \subseteq \calS$ for $i \in [7]$, and
  let $\mu_i$ denote the number of candidates that are not contained in any triple in~$\calS_i$, i.e., 
  $\mu_i=3k-3\nu_i$.
  
  We form an election $E = (C,V)$ with the following
  candidates:
  \begin{enumerate}
  \item For each $x_i \in X$, we have a corresponding candidate that
    we also denote with $x_i$. We refer to these candidates as the
    \emph{original} ones.
  \item For each~$i \in [7]$, we have $2(\nu_i+\mu_i)$ sets of
    candidates, $A_i^1, \ldots, A_i^{\nu_i}$, 
    $\widetilde{A}_i^1, \ldots, \widetilde{A}_i^{\mu_i}$, 
    $B_i^1, \ldots, B_i^{\nu_i}$, 
    $\widetilde{B}_i^1, \ldots, \widetilde{B}_i^{\mu_i}$, each containing exactly three
    candidates. We refer to these sets as the \emph{dummy triples} and
    to their members as the \emph{dummy candidates}.
  \end{enumerate}
  For each $i \in [7]$, we introduce two voters, so we will have $n=14$ voters. 
  Let us rename the $\mu_i$ original candidates not contained in any triple from~$\calS_i$ so that
  $X \setminus (\bigcup \calS_i) = \{y_1, \ldots, y_{\mu_i}\}$. The exact details of
  this renaming are irrelevant for the correctness of our construction. We form voters~$v_i$ and~$u_i$ with the following
  preference orders:
  \begin{align*}
    v_i \colon  & S_i^1  \pref A_i^1 \pref S_i^2 \pref A_i^2 \pref \cdots \pref S_i^{\nu_i} \pref A_i^{\nu_i} \\
    & \pref y_1 \pref \widetilde{A}_i^1 \pref y_2 \pref \widetilde{A}_i^{2} \pref \cdots \pref y_{\mu_i} \pref \widetilde{A}_i^{\mu_i} 
    \pref [\cdots]; \\
    u_i \colon  & S_i^1 \pref B_i^1 \pref S_i^2 \pref B_i^2 \pref \cdots \pref S_i^{\nu_i} \pref B_i^{\nu_i} \\
    & \pref y_1 \pref \widetilde{B}_i^1 \pref y_2 \pref \widetilde{B}_i^{2} \pref \cdots \pref y_{\mu_i} \pref  \widetilde{B}_i^{\mu_i} 
    \pref [\cdots]
  \end{align*}
  where by $[\cdots]$ we mean listing the remaining dummy candidates
  in an arbitrary order, except that members of each dummy triple are
  ranked consecutively. 
  We form an instance of \textsc{Independent-Clone Partition} with
  election $E = (C,V)$, where we ask if there is a partition of $C$
  into at most $t = \nicefrac{|C|}{3}$ independent clones, each of
  size at most $b = 3$, and recognized by at least $w = 2$ voters.
  \smallskip

  Let us assume that there is a subfamily $\calS'$ of $\calS$ that
  contains exactly~$k$ sets whose union is~$X$. One can verify that
  this corresponds to an independent-clone partition of~$E$ where each
  dummy triple is a size-$3$ clone (recognized by all the voters) and
  each set $S \in \calS'$ is a size-$3$ clone (recognized by voters~$v_i$ and~$u_i$ where $S \in \calS_i$).
  
  For the other direction, assume that there is a clone partition that
  meets the requirements of our instance of \textsc{Independent-Clone Partition}. Since
  $t~=~\nicefrac{|C|}{3}$ and $b = 3$, we know that this partition
  contains exactly~$t$ independent clones of size~$3$ such that each
  clone is recognized by at least $w = 2$ voters. By our construction,
  none of these clones consists of both original and dummy candidates:
  indeed, for each original candidate $x_i$ and each dummy candidate
  $d$, there is at most one voter $z$ for whom
  $|\pos_z(x_i)-\pos_z(d)| \leq 2$. Next, we note that for each clone
  $\{x_a, x_b, x_c\}$ from our partition, the only two voters that
  rank its members consecutively must be $v_i$ and~$u_i$ for some $i \in [7]$ with
  $S= \{x_a,x_b, x_c\} \in \calS_i$. Consequently, clones that
  consist of the original candidates correspond to a solution to the
  input \textsc{RX3C} instance.
  
\end{proof}

\thmindepclonepartWhardt*
\begin{proof}
    We present a reduction from a special variant of the \textsc{Perfect Code} problem. In \textsc{Perfect Code}, we are given an undirected graph~$G=(U,F)$ and an integer~$k$ and the task is to decide whether there exists a \emph{perfect code of size~$k$}, that is, a set~$S$ of~$k$ vertices in~$G$ such that each vertex~$u \in U$ is contained in the closed neighborhood of exactly one vertex from~$S$. 
    Here, the \emph{closed neighborhood} of some vertex~$u \in U$, denoted by~$N[u]$, contains~$u$ together with all vertices of~$G$ adjacent to~$u$. 
    This problem is \Wh{1} when parameterized by~$k$ \cite{dow-fel:j:fixed-parameter-I}. 

    \paragraph{Multicolored Perfect Code.}
    It will be useful for us to use the following variant of \textsc{Perfect Code} that we call \textsc{Multicolored Perfect Code}: given an undirected graph~$G=(U,F)$ with its vertex set partitioned into $k$ sets~$U_1,\dots,U_k$, find a \emph{multicolored} perfect code in~$G$, i.e., a perfect code that contains exactly one vertex from each set~$U_i$, $i \in [k]$. 
    To see that this problem is \Wh{1}, we present a simple parameterized reduction from \textsc{Perfect Code}. 
    Given an input graph $G=(U,F)$ and parameter~$k$, we create a graph~$G'=(U',F')$ by introducing $k$ copies $u_1,\dots,u_k$ of each vertex $u \in U$, with $U'$ partitioned into the sets $U_i=\{u_i:u \in U\}$ for $i \in [k]$, and we add an edge between two vertices of~$U'$ if and only if either they are two copies of the same vertex in~$U$, or they are copies of two adjacent vertices in~$U$.
    It is straightforward to verify that there is a perfect code of size~$k$ in~$G$ if and only if there is a multicolored perfect code in~$G'$.


    \paragraph{Construction.}    
    Consider an instance of \textsc{Multicolored Perfect Code} with input graph $G=(U,F)$ with its vertex set partitioned into~$k$ sets $U_1,\dots,U_k$. Let $U=\{u_1,\dots,u_r\}$.
    %
    We construct an election~$(C,V)$ as follows.  
    For each vertex $u_j \in U$, we create a corresponding \emph{original candidate} also denoted as~$u_j$.
    For each $i \in [k]$ we create two \emph{terminal candidates} $s_i$ and~$t_i$.
    Additionally, we create eight sets $A, \widetilde{A}, B,\widetilde{B}, D, \widetilde{D}, E, \widetilde{E}$ of dummy candidates, each of size~$r$. 
    We let $A=\{a_1,\dots,a_r\}$, and we use the analogous notation for all the other seven dummy candidate sets; subscripts of dummies are interpreted modulo~$r$ throughout the proof.
    Let us also denote the set of terminal candidates by $\{y_1,\dots,y_{2k}\}$.
    
    For each vertex $u_j \in U_i$ for some~$i \in [k]$, we create two voters, $v_j$ and~$w_j$ whose preferences are as follows:
    \begin{align*}
        v_j \colon & s_i \pref N[u_j] \pref t_i \pref b_j \pref B\setminus \{b_j\} \pref \widetilde{B} \pref E \pref \widetilde{E}  \\
        & \pref a_j \pref \maybe{u_1} \pref \widetilde{a}_{j} \pref a_{j+1} \pref \maybe{u_2} \pref \widetilde{a}_{j+1} 
        \\ & 
        \pref \cdots \pref  a_{j+r} \pref \maybe{u_r} \pref \widetilde{a}_{j+r} \\
        & \pref d_j \pref \maybe{y_1} \pref \widetilde{d}_{j} \pref d_{j+1} \pref \maybe{y_2} \pref \widetilde{d}_{j+1} 
        \\ & 
        \pref \cdots \pref  d_{j+2k} \pref \maybe{y_{2k}} \pref \widetilde{d}_{j+2k} \pref [\cdots], \\
        w_j \colon & s_i \pref N[u_j] \pref t_i \pref a_j \pref A\setminus \{a_j\} \pref \widetilde{A} \pref D \pref \widetilde{D}  \\
        & \pref b_j \pref \maybe{u_1} \pref \widetilde{b}_{j} \pref b_{j+1} \pref \maybe{u_2} \pref \widetilde{b}_{j+1} 
        \\ & 
        \pref \cdots \pref  b_{j+r} \pref \maybe{u_r} \pref \widetilde{b}_{j+r} \\
        & \pref e_j \pref \maybe{y_1} \pref \widetilde{e}_{j} \pref e_{j+1} \pref \maybe{y_2} \pref \widetilde{e}_{j+1}
        \\ & 
        \pref \cdots \pref  e_{j+2k} \pref \maybe{y_{2k}} \pref \widetilde{e}_{j+2k} \pref [\cdots]
    \end{align*}
    where $\maybe{c}$ for some candidate~$c$ should be interpreted simply as~$c$ whenever $c$ does not appear earlier in the preference list, and should be left out when it does; 
    we use $[\cdots]$ to list all remaining dummy candidates in arbitrary order.
    We will refer to the candidates preceding the first dummy in these votes ($b_j$ or~$a_j$) as the \emph{main block} of the vote.

    We create an instance of \textsc{Independent-Clone Partition} with election~$(C,V)$ by setting $w=2$, $t=k+8$, and $b=r$. Note that this indeed yields a parameterized reduction with parameter~$t$.

    \paragraph{Correctness.}
    Les us assume first that there is a multicolored perfect clique~$X$ to our instance of \textsc{Perfect Code}. Then $\{N[u]:u \in X \}$ is a partitioning of the vertex set~$U$, and hence, $\{N[u] \cup \{s_i,t_i :u \in X, U_i \ni u\}$ is a partitioning of all non-dummy candidates. Adding the eight dummy candidate sets creates an independent-clone partition for~$E$ that contains $t=k+8$ clones, each of size at most~$r$, with each clone of the form $N[u_j] \cup \{s_i,t_i\}$ for some $u_j \in X \cap U_i$ recognized by voters~$v_j$ and~$w_j$, and with each clone containing dummies recognized by half of the voters.

    Assume now that there is an independent-clone partition~$\Pi$ that satisfies all requirements of our instance. 
    The key observation is that no non-dummy candidate is placed next to the same dummy vertex in more than one vote, due to the cyclic shift between dummies. 
    By $w=2$, this means that all independent clones in~$\Pi$ that contain a non-dummy candidate must be recognized by some voter as occurring within the main block (as after the main block, non-dummy candidates are always sandwiched between dummies). Consequently, a clone in~$\Pi$ containing $s_i$ or~$t_i$ must be recognized by voters in~$V^\star_i=\{v_j,w_j: u_j \in U_i\}$. In particular, no voter can recognize a clone that contains terminal candidates both from $\{s_i,t_i\}$ and $\{s_j,t_j\}$ for some indices $i \neq j$.
    Thus, there must be at least~$k$ clones in~$\Pi$ containing terminal candidates but no dummies. Since placing the~$8$ set of dummies into clones necessitates at least~$8$ clones (as each dummy set has size~$r=b$), by $t=k+8$ we get that $\Pi$ must contain exactly $k$ clones that contain non-dummy candidates, with one clone, say $\pi_i$, containing both~$s_i$ and~$t_i$ for each $i \in [k]$. 

    Observe now that the clones $\pi_1,\dots,\pi_k$ also partition all the original candidates. Moreover, by $s_i,t_i \in \pi_i$, we know that $\pi_i$ can only be recognized by some voter $v_j$ or $w_j$ in~$V^\star_i$ if the main block of this voter coincides with~$\pi_i$, that is, $\pi_i=N[u_j] \cup \{s_i,t_i\}$. 
    This means that each set $\pi_i \setminus \{s_i,t_i\}$ is of the form $N[u_j]$ for some vertex~$u_j \in U_i$; let us select such a vertex $u_j$. Then the sets $\pi_i \setminus \{s_i,t_i\}$ for $i \in [k]$ form a partitioning of~$U$ and the selected vertices form a multicolored perfect code for~$G$.
    This shows the correctness of our reduction.
    
\end{proof}

\section{Missing Proofs from Section~\ref{sec:subelection}}

\thmsubelectinbnpc*
\begin{proof}
  To see that the problem is in $\np$, it suffices to guess a
  subelection and its perfect-clone partition, and verify that they
  meet the requirements of the given instance.

  Next, we give a reduction from \textsc{RX3C}.  Let $(X,\calS)$
  be our input instance where $X = \{x_1, \ldots, x_{3k}\}$ is a set
  of $3k$ elements and $\calS = \{S_1, \ldots, S_{3k}\}$ is a family
  of size-$3$ subsets of $X$. We form an election $E = (C,V)$ where
  \[
    C = \{ s_1, s^+_1, s^-_1, \ldots, s_{3k}, s^+_{3k}, s^-_{3k}\}
  \]
  and $V$ contains following sets of voters:
  \begin{enumerate}
  \item There are $9k^2$ \emph{type-1 voters}, each with preference order:
    \[
       s^+_1 \pref s_1 \pref s^-_1 \pref s^+_2 \pref s_2 \pref s^-_2 \pref \cdots
    \]
    and $9k^2$ \emph{type-2 voters}, each with
    preference order:
    \[
       s^-_1 \pref s_1 \pref s^+_1 \pref s^-_2 \pref s_2 \pref s^+_2 \pref \cdots
    \]

  \item For each element $x_i \in X$ and each set $S_j$ that contains~$x_i$, we form a \emph{coverage voter} $v_{i,j}$ whose preference
    order is the same as that of type-1 voters, except that for each
    $\ell \in [3k]$ we have the following modifications: 
    \begin{itemize}
        \item $v_{i,j} \colon s^-_j \pref s^+_j \pref s_j$, and 
        \item if
    $S_\ell$ contains $x_i$ but $\ell \neq j$, then
    $v_{i,j} \colon s^+_\ell \pref s^-_\ell \pref s_\ell$.
    \end{itemize} 
      There are at most
    $9k^2$ coverage voters.

  \end{enumerate}
  We set the upper bound on the clone sizes to be $b=2$, the maximum
  number of clones to be $t = {6k}$, and the number of voters in the
  subelection to be $w = 18k^2+3k$.

  Let us assume that there is a subelection $E' = (C,V')$ with at
  least $w = 18k^2 + 3k$ voters that has a perfect-clone partition
  $(C_1, \ldots, C_{t'})$, such that $t' \leq t = 6k$ and for each
  $i \in [t']$ we have $|C_i| \leq 2$.  By a simple counting argument,
  $E'$ must include some type-1 and some type-2 voters and, hence, we
  can assume it includes all of them. Due to the preference orders of
  type-1 and type-2 voters, for each set~$S_j$, we either have clones
  $\{s_j, s_j^+\}$ and $\{s^-_j\}$, or $\{s_j, s_j^-\}$ and
  $\{s^+_j\}$, or $\{s_j\}, \{s^+_j\}, \{s^-_j\}$. However, this last
  option is impossible as we have at most $6k$ clones. Intuitively, if
  we have clone $\{s_j, s^+_j\}$ then $S_j$ is included in a solution
  for the \textsc{RX3C} instance, and if we have clone
  $\{s_j, s^-_j\}$ then it is not.

  Next, we observe that if $E'$ includes some voter $v_{i,j}$ then,
  due to $v_{i,j}$'s preference order and preceding paragraph, we must
  have clone $\{s_j, s^+_j\}$ and for every
  $\ell \in [3k] \setminus \{j\}$ such that $x_i \in S_\ell$, we must
  have clone $\{s_\ell, s^-_\ell\}$. This means that for each
  $i \in [3k]$, $E'$ contains at most one coverage voter for
  $x_i$. However, since $E'$ contains $18k^2+3k$ voters, it must
  contain one coverage voter for each $x_i \in X$. This means that
  sets $S_\ell$ for which we have clones $\{s_\ell, s^+_\ell\}$ must
  form an exact cover of $X$.

  For the other direction, assume that $(X,\calS)$ is a yes-instance of the \textsc{RX3C} problem and that
  there are $k$ sets from $\calS$ whose union is $X$. For simplicity,
  let us assume that these are $S_1, \ldots, S_k$. We form a clone
  partition that contains clones $\{s_i, s^+_i\}$
  and $\{s^-_i\}$ for each $i \in [k]$, and clones $\{s_\ell, s^-_\ell\}$ and $\{s^+_\ell\}$ for each $\ell \in [3k] \setminus [k]$. We form a subelection
  $E'$ that includes all type-1 voters, all type-2 voters, and for
  each $x_i$ we include a single coverage voter $v_{i,j}$ such that
  $j \in [k]$ and $x_i \in S_j$ (this is possible since
  $S_1, \ldots, S_k$ form an exact cover of $X$). This shows that we
  have a yes-instance of our \textsc{Subelection-Clone Partition}
  instance.
\end{proof}

\thmsubelectionwhard*
\begin{proof}
    We present a reduction from the \Wh{1} \textsc{Multicolored Clique} problem where, given an undirected graph $G=(U,F)$ and an integer parameter~$k$ with the vertex set of~$G$ partitioned into~$k$ independent sets~$U_1,\dots,U_k$, the task is to decide if $G$ contains a clique of size~$k$.
    We may assume that $|U_i|=\dots=|U_k|=r$ for some integer~$r$, so that we can write $U_i=\{u^h_i: h \in [r]\}$ for each $i$.

    We fix integers $\Delta=k(r-1)$ and $R=\Delta|F|$. 
    To define the candidate set~$C$, we introduce $D_i=\{d^h_i:h \in [R]\}$ of dummy candidates for each $i \in [k-1]$;
    throughout this proof, dummies' superscripts are modulo~$|R|$. 
    Additionally, we create candidate sets~$A_i=\{a^h_i:h \in [R-r]\}$, $B_i=\{b^h_i:h \in [R-r]\}$, and $C_i=\{c_i^h:h \in [r+1]\}$ for each $i \in [k]$. 
    The total number of candidates is thus 
    \begin{equation}
    \label{eq:sizeofC}
    \begin{split}
        |C|&=(k-1)R+k(2(R-r)+(r+1))\\
        & =(3k-1)R-\Delta.
    \end{split}
    \end{equation}
    
    Next, we create a voter~$v_e$ for each edge $e \in F$, so the number of voters is $n=|F|$. 
    
    To define the preferences of the voters, we will fix an arbitrary bijection~$\varphi:F \rightarrow [n]$. 
    The preferences of voter~$v_e$ for some edge $e=\{u_p^i,u_q^j\}$ where $i<j$ are defined as a series of blocks. 
    First, for each index $h \in [k]$, we define a \emph{vertex block $X_h$}. 
    Second, we define $m$ \emph{dummy blocks}, each containing~$R$ dummy candidates: 
    the $j^{\textrm{th}}$ dummy block orders dummies in~$D_j$ using a cyclical shift, starting with the $(\shift(j)+1)^\textrm{th}$ and ending with the $(\shift(j)+R)^\textrm{th}$ dummy candidate (modulo $R$) where $\shift(e)=(\varphi(e)-1)\Delta$ is a shifting value that depends on the number~$\varphi(e)$ corresponding to the edge~$e$. Formally, we define these blocks as follows.
    \begin{align*}
        X_{h} &: \ora{A_h} \succ c_h^1 \succ \dots \succ c_h^{r+1} \succ \ora{B_h} \qquad \text{ for }h \in [k] \setminus \{i,j\}; \\
        X_{i} &:c_i^1 \succ \dots \succ c_i^p \succ \ora{A_i} \succ \ora{B_i} \succ c_i^{p+1} \succ \dots \succ c_i^{r+1}; \\
        X_{j} &:c_j^1 \succ \dots \succ c_j^q \succ \ora{A_j} \succ \ora{B_j} \succ c_j^{q+1} \succ \dots \succ c_j^{r+1}; \\
        D_h^e &: d_h^{\shift(e)+1} \succ \dots \succ d_h^{\shift(e)+R} 
    \end{align*}
    where candidates within each set $\ora{A_h}$ and~$\ora{B_h}$ are ordered increasingly according to their superscripts. 
    
    We are now ready to describe the preferences of voter~$v_e$:
    \begin{align*}
        v_e &: X_1 \succ D_1^e \succ \dots 
        \succ X_{k-1} \succ D_k^e \succ X_k.
    \end{align*}
    
    We finish our construction by setting $w=\binom{k}{2}$, $t=3k-1$, and  $b=R$.
    Observe that the presented reduction is a parameterized reduction with parameter~$w+t$.

    We claim that the constructed election~$E=(C,V)$ with integers~$b$, $w$, and~$t$ forms a yes-instance of \textsc{Subelection-Clones Partition} if and only if $G$ contains a clique of size~$k$.

    Suppose first that there exists a subelection~$E'=(C,V')$ of~$E$ and a partitioning~$\Pi$ of the candidate set~$C$ into at most $t=3k-1$ sets, each of them a perfect clone in~$E'$ of size at most~$b$. Note that by (\ref{eq:sizeofC}) we know
    \[(t-1)b<|C| =tb-k(r-1), \]
    which implies that $\Pi$ must consist of  exactly $t$ subelection clones, and moreover, their total size is only $\Delta=R/n$ smaller than~$tb$. We will argue that each dummy block forms a clone in~$\Pi$, and each vertex block is the union of two  clones from~$\Pi$.
    
    To see this, suppose that some voter~$v_e$ is contained in the set~$V'$ of voters who recognize all sets in~$\Pi$ as perfect clones. Let $\pi_1,\dots,\pi_t$ be the clones in~$\Pi$ ordered according to~$\succ_{v_e}$; then $\pi_j$ for some $j \in [t]$ 
    starts 
    at earliest on the 
    $((j-1)b-\Delta+1)^{\textrm{th}}$ position
    and
    ends at latest on the $(jb)^{\textrm{th}}$ position. 
%
    Fix an index $i \in [k-1]$.
    By the previous paragraph, the subelection clone $\pi_{3i-1}$, containing mostly candidates from block~$X_i$, may contain at most the first $\Delta$ dummies from the dummy block~$D_i^{e}$. However, due to the cyclical shift of the dummy candidates within a dummy block (a shift by $\Delta$ from one vote to the next), the set of the first $\Delta$ dummy candidates within $D_i^e$ and~$D_i^{e'}$ are disjoint whenever $e \neq e'$, due to $|D_i|=\Delta |F|$. Consequently, as $w>1$, we get that $\pi_{3i-1}$ cannot contain any candidates from~$D_i$. 
    Similarly, the subelection clone $\pi_{3i+1}$, containing  mostly candidates from block~$X_{i+1}$, may contain at most the last $\Delta$ dummies from the dummy block~$D_i^{e}$. Using again the same argument based on the cyclical shift within each dummy block, we get that $\pi_{3i+1}$ cannot contain any candidates from~$D_i$. 
    As a consequence, the subelection clone $\pi_{3i}$ must comprise \emph{exactly} the dummy block~$D_i^e$, i.e., $\pi_{3i}=D_i$ for each $i \in [k-1]$.
    It follows that each vertex block~$X_i$ for some $i \in [k]$ must be the union of two subelection clones in~$\Pi$.
    In particular, these two blocks must have the form $A_i$

    Recall our assumption that $v_e$ is contained in voter set~$V'$ of the subelection~$E'$ that recognizes all clones in~$\Pi$, and let $e=\{u_i^p,u_j^q\}$ for some $i<j$. 
    Since both $\pi_{3i-2}$ and $\pi_{3i-1}$ must have size between $b-\Delta=R-k(r-1)$ and~$b=R$, we know that 
    $\pi_{3i-2} \cap C_i=\{c_i^1,\dots,c_i^p\}$. 
    Similarly, we get that 
    $\pi_{3j-2} \cap C_j=\{c_j^1,\dots,c_j^q\}$. 
    We refer to this fact that \emph{voter~$v_e$ selects vertices~$u_i^p$ and~$u_j^q$}. 
    Since the above reasoning holds for all voters in~$V'$, it follows that every voter in~$V'$ selects the end-vertices of the corresponding edge. Hence, the set~$F^\star =\{e:v_e \in V'\}$ has the property that all edges in~$F^\star$ that are incident to some vertex in~$U_i$, $i \in [k]$, must select the same vertex, i.e., must be incident to the same vertex of~$V_i$. Hence, there can be at most $k$ selected vertices in~$G$, and they must contain all end-vertices of the $\binom{k}{2}$ edges in~$F^\star$, implying that the selected vertices form a clique of size~$k$ in~$G$.
    
    Suppose now for the other direction that $K=\{u_i^{\sigma(i)}:i\in [k]\}$ form a clique in~$G$. Let $V'$ be the set of voters that correspond to edges contained in the clique induced by~$K$. It is then straightforward to check that the partitioning~$\Pi$ that contains the following clones satisfies all requirements:
    \begin{itemize}
        \item a clone~$D_i$ for each $i \in [k-1]$, and 
        \item clones $A_i \cup \{c_i^1,\dots,c_i^{\sigma(i)}\}$ and $\{c_i^{\sigma(i)+1},\dots,c_i^{r+1}\} \cup B_i$ for each $i \in [k]$.
    \end{itemize}
    In particular, each set in~$\Pi$ has size at most $R=b$, $|\Pi|=3k-1$, and $\Pi$ forms a perfect-clone partition for each voter in~$V'$.
    This proves the correctness of our reduction.
    
\end{proof}

\end{appendices}
\fi

\end{document}